\documentclass[journal]{IEEEtran}

\usepackage{cite}
\usepackage{graphicx} 
\usepackage{epsf}
\usepackage{subfigure}
\usepackage{amsmath}
\usepackage{comment}
\usepackage{amssymb}
\usepackage{array}
\usepackage{setspace}
\usepackage{amsthm,amssymb}
\usepackage[ruled,lined]{algorithm2e}
\usepackage{multirow} 
\usepackage{tabularx}
\usepackage{xfrac}
\usepackage{breqn}
\usepackage{bbm}
\usepackage{enumerate}
\usepackage{algorithmic}
\usepackage{dsfont}
\usepackage{float}
\usepackage{mathtools}
\usepackage{color}

\newtheorem{prop}{Proposition}

\graphicspath{{fig/}}

\begin{document}

\title{\huge Downlink Rate Maximization with Reconfigurable Intelligent Surface Assisted Full-Duplex Transmissions}


\author{ \IEEEauthorblockN{Li-Hsiang Shen$^\dagger$, Chia-Jou Ku, and Kai-Ten Feng}\\
\IEEEauthorblockA{$^\dagger$Department of Communication Engineering, National Central University, Taoyuan 30021, Taiwan\\
Department of Electronics and Electrical Engineering, \\ National Yang Ming Chiao Tung University, Hsinchu 300093, Taiwan\\
gp3xu4vu6@gmail.com, chiajouku.ee09@nycu.edu.tw, and ktfeng@nycu.edu.tw}}

\maketitle

\begin{abstract}
Reconfigurable intelligent surfaces (RIS) as an effective technique for intelligently manipulating channel paths through reflection to serve desired users. Full-duplex (FD) systems, enabling simultaneous transmission and reception from a base station (BS), offer the theoretical advantage of doubled spectrum efficiency. However, the presence of strong self-interference (SI) in FD systems significantly degrades performance, which can be mitigated by leveraging the capabilities of RIS. In this work, we consider joint BS and RIS beamforming for maximizing the downlink (DL) transmission rate while guaranteeing uplink (UL) rate requirement. We propose an FD-RIS beamforming (FRIS) scheme by adopting penalty convex-concave programming. Simulation results demonstrate the UL/DL rate improvements achieved by considering various levels of imperfect CSI. The proposed FRIS scheme validates their effectiveness across different RIS deployments and RIS/BS configurations. FRIS has achieved the highest rate compared to the other approximation method, conventional beamforming techniques, HD systems, and deployment without RIS.
\end{abstract}




\section{Introduction}
Reconfigurable intelligent surface (RIS) \cite{RIS_intro_princi, acm} introduces a promising hardware to assist advanced network transmissions. An RIS is composed of passive elements constructed from massive meta-material based structures \cite{RIS_intro_hardware}. Each of these reflecting element can adjust its phase, enabling the creation of diverse multipath profiles and redirection of impinging signals towards the desired direction. Compared to traditional beamforming \cite{BM_MRC}, RIS has its potential advantage to reckon with insufficient degree of freedom of channels. Benefited from its high energy/spectral efficiency as well as cost-effective deployment, RIS contribute significantly to the improved signal quality and extended coverage. The optimization of RIS phase for maximizing energy efficiency is addressed in \cite{RIS_energy_efficiency}. Additionally, the work presented in \cite{RIS_power} aims for minimizing the overall power consumption of the base station (BS) by dynamically adjusting the configuration of RIS. In \cite{ris_secret}, RIS is dedicated to the design and optimization specifically for enhancing the physical layer security, preventing information leakage to eavesdroppers. In \cite{BM_HD}, a downlink based transmission is improved with the employment of RIS. In \cite{SDR}, an advanced power control for broadcasting is designed with the aid of RIS deployment.

Moreover, employing half-duplex (HD) transmission \cite{DLMIMO} has the potential to induce the inefficient utilization of limited spectrum resources. As a remedy, full-duplex (FD) communication systems have emerged as a solution, which allows simultaneous transmission and reception of both DL and UL signals on the same frequency band. FD can therefore theoretically double the spectrum efficiency compared to HD systems. Nevertheless, the substantial challenge in FD systems arises from the presence of strong self-interference (SI), which can substantially degrade system performance. A study presents a methodology by designing based on oblique projection to mitigate SI in null space \cite{FD_oblique, FD_queue}. However, this technique requires stringent assumptions about the BS owing to the limited diversity of antennas it is equipped with. Hence, the integration of RIS and FD transmission offers an alternative solution, as RIS can effectively mitigate both SI and inter-user interferences between uplink (UL) and downlink (DL) transmissions by employing destructive signals along reflective paths. In \cite{ris_fd_relay}, relaying FD signals are accomplished with the assistance of RIS. In \cite{ris_fd_ra}, resource allocation for spectral-efficient RIS-FD is conceived in non-orthogonal multiple access. In \cite{ris_fd_power}, RIS-enabled power minimization is considered with RIS possessing both reflective and refractive functionalities. The work in \cite{amy} has derived the closed-form presenting the relationship between the number of receiving antennas and that of RIS elements but under no BS beamforming. In this study, we explore the arising potential when jointly considering BS and RIS beamforming in FD systems. The contributions of this work can be summarized as follows.
\begin{itemize}
\item
We formulate a DL rate maximization problem while guaranteeing UL quality-of-service (QoS), i.e., minimum UL rate requirement for RIS-FD systems. We consider the constraints of BS beamforming as well as RIS phase-shift configuration.

\item
We propose an FD-RIS beamforming (FRIS) scheme by employing the Lagrangian dual and Dinkelbach transformation to address logarithmic and fractional functions. The original problem is decoupled into the two subproblems, which are  solved by successive convex approximation (SCA) and penalty convex-concave programming (PCCP), respectively.
 
\item 
Simulation results demonstrate the individual UL/DL rates with different QoS levels. The optimal deployment of the RIS with different numbers of RIS elements is also shown. FRIS outperforms the existing methods in terms of existing approximation technique, deployment without RIS, conventional beamforming, and half-duplex systems.
\end{itemize}

\section{System Model and Problem Formulation} \label{SYS_MOD}

\subsection{Signal Model}

We consider an FD-BS and an RIS, serving $M$ DL user equipments (UEs) as well as $N$ UL UEs, with respective sets as $\mathcal{M}=\{1,...,M\}$ and $\mathcal{N}=\{1,...,N\}$. Note that UEs equipped with a single antenna operate in HD mode, i.e., either UL or DL. The BS is equipped with $\mathcal{N}_t=\{1,...,N_t\}$ DL transmit antenna set (Tx) and $\mathcal{N}_r=\{1,...,N_r\}$ receiving UL antenna set (Rx), whereas RIS has $\mathcal{K}=\{1,...,K\}$ reflecting element set. The UL channels of UL UE-BS, UL UE-RIS, and RIS-BS are defined as ${\rm\mathbf{U}} \in \mathbb{C}^{N_r \times N}$, ${\rm\mathbf{U}}_1 \in \mathbb{C}^{K \times N}$ and ${\rm\mathbf{U}}_2 \in \mathbb{C}^{N_r \times K}$, respectively. The DL channels of BS-DL UE, BS-RIS, and RIS-DL UE are denoted as ${\rm\mathbf{D}} \in \mathbb{C}^{M \times N_t}$, ${\rm\mathbf{D}}_1 \in \mathbb{C}^{K \times N_t}$ and ${\rm\mathbf{D}}_2 \in \mathbb{C}^{M \times K}$, respectively. Notation ${\rm\mathbf{S}} \in \mathbb{C}^{N_r \times N_t}$ denotes the SI channel from BS Tx to itself Rx, whilst ${\rm\mathbf{V}} \in \mathbb{C}^{M \times N}$ presents the co-channel interference (CCI) from UL to DL UEs. As for RIS, we define RIS reflection matrix as ${\boldsymbol{\Theta}}={\rm{diag}} ( [\beta_{k}e^{j\theta_k}] ) \in \mathbb{C}^{K \times K}$, where $\beta_k$ is the RIS amplitude and $e^{j\theta_k}$ is the RIS phase where $\theta_k \in [ 0,2\pi) \forall k\in \mathcal{K}$. We consider $p_D$ and $p_U$ as transmit power of DL BS and UL UE, respectively. We define ${\rm\mathbf{W}} = [{\rm\mathbf{w}}_1,...,{\rm\mathbf{w}}_M] \in \mathbb{C}^{N_t \times M} $ as the precoding matrix for BS beamforming. As for signals, ${\rm\mathbf{x}}_D = [x_{D,1},...,x_{D,M}]^{T} \in \mathbb{C}^{M} $ and ${\rm\mathbf{x}}_U=[x_{U,1},...,x_{U,N}]^{T} \in \mathbb{C}^{N}$ are respectively defined as the DL and UL user data streams. Therefore, the received DL signal of DL UE $m$ is given by
\begingroup
\allowdisplaybreaks
\begin{align}\label{eqn:DL y_m}
    y_m &= \sqrt{p_D}\left( {\rm\mathbf{D}}_m + {\rm\mathbf{D}}_{2,m}{\boldsymbol{\Theta}}{\rm\mathbf{D}}_{1} \right){\rm\mathbf{w}}_m x_{D,m}\notag \\
    &+ \sqrt{p_U}({\rm\mathbf{V}}_m +{\rm\mathbf{D}}_{2,m}{\boldsymbol{\Theta}}{\rm\mathbf{U}}_{1}){\rm\mathbf{x}}_{U} \notag \\
    & +\! \sum_{m' \in \mathcal{M}\backslash m}   \sqrt{p_D}({\rm\mathbf{D}}_m + {\rm\mathbf{D}}_{2,m}{\boldsymbol{\Theta}}{\rm\mathbf{D}}_{1}){\rm\mathbf{w}}_{m'} x _{D,m^{'}} + n_m.
\end{align}
\endgroup
In $\eqref{eqn:DL y_m}$, first term represents downlink signal, second one indicates CCI, and third equation stands for inter-DL UE interference. Notations $\mathbf{D}_m \in \mathbb{C}^{1 \times N_t}$ and $\mathbf{D}_{2,m}\in \mathbb{C}^{1 \times K}$ denote the $m$-th row vector in $\mathbf{D}$ and $\mathbf{D}_{2}$. Also, ${\rm\mathbf{V}}_m\in \mathbb{C}^{1 \times N}$ indicates CCI from UL to DL UE $m$. $T$ is the transpose operation. Notation $n_m \sim \mathcal{CN}(0,\sigma_m^2)$ is the additive white Gaussian noise (AWGN) with zero mean and variance $\sigma_m^2$. Moreover, the received UL signal at BS Rx is expressed as
\begin{align}\label{eqn:UL y_U}
    {\rm\mathbf{y}}_U \!=\! \sqrt{p_U}({\rm\mathbf{U}}+{\rm\mathbf{U}}_2 \boldsymbol{\Theta} {\rm\mathbf{U}}_1){\rm\mathbf{x}}_U \!+\! \sqrt{p_D}({\rm\mathbf{S}}+{\rm\mathbf{U}}_2 \boldsymbol{\Theta} {\rm\mathbf{D}}_1){\rm\mathbf{W}}{\rm\mathbf{x}}_D \!+\! {\rm\mathbf{n}}_U,
\end{align}
where ${\rm\mathbf{n}}_U \sim \mathcal{CN}({\rm\mathbf{0}}, \sigma_U^2 {\rm\mathbf{e}}_{N_r})$ is AWGN noise. ${\rm\mathbf{0}}$ is zero-mean vector with length $N_r$, $\sigma_U^2$ indicate UL noise power, and ${\rm\mathbf{e}}_{N_r}$ is the unit vector with length $N_r$. As shown in the second term of $\eqref{eqn:UL y_U}$, SI comes solely from BS Tx to its Rx. Compared to non-RIS network, we can observe additional reflective interferences in both CCI and SI. Furthermore, $\mathbf{\Theta}$ is coupled between the cascaded channels, which can be solved by \cite{replace}. For brevity, we only reformulate ${\rm\mathbf{D}}_{2,m}{\boldsymbol{\Theta}}{\rm\mathbf{D}}_{1} = \boldsymbol{\theta} {\rm{diag}}({\rm\mathbf{D}}_{2,m}) {\rm\mathbf{D}}_{1} $ with $\boldsymbol{\theta}$ as vectorized RIS configuration; while the others are omitted here. Accordingly, the signal-to-interference-plus-noise ratio (SINR) of DL UE $m$ and UL BS Rx for UL UE $n$ are respectively expressed as
\begingroup
\allowdisplaybreaks
\begin{align}
    \text{SINR}_{D,m}
 	&=
 	\frac{D_{S,m}({\boldsymbol{\Theta}}, {\rm\mathbf{W}})}{ D_{C,m}({\boldsymbol{\Theta}}) + D_{I,m}({\boldsymbol{\Theta}}, {\rm\mathbf{W}}) + \sigma^2},\label{eqn: DL_R} \\
    \text{SINR}_{U,n} 
    &= 
    \frac{U_{S,n}({\boldsymbol{\Theta}})}{U_{I,n}({\boldsymbol{\Theta}}, {\rm\mathbf{W}}) + \sigma^2}, \label{eqn: UL_R}
\end{align}
\endgroup
where $\overline{x}_{D,m} = \sqrt{p_D}x_{D,m}, \forall m\in\mathcal{M}$, $\overline{x}_{U,n} = \sqrt{p_U}x_{U,n}, \forall n\in\mathcal{N}$, $\overline{\rm\mathbf{x}}_{D} = \sqrt{p_D} {\rm\mathbf{x}}_{D}$, and $\overline{\rm\mathbf{x}}_{U} = \sqrt{p_U} {\rm\mathbf{x}}_{U}$. For brevity, we omit the definition of respective terms in $\eqref{eqn: DL_R}$ and $\eqref{eqn: UL_R}$, as those parameters can be formed by applying a norm operation to signals or interferences. Noise power for all UEs is denoted as $\sigma^2$. The DL/UL user data rate can be respectively attained as
\begingroup
\allowdisplaybreaks
\begin{align}
	R_{D,m}({\boldsymbol{\Theta}}, {\rm\mathbf{W}})&=  \log_2({1+\text{SINR}_{D,m}}), \label{rate_ob1}\\
	R_{U,n}({\boldsymbol{\Theta}}, {\rm\mathbf{W}}) &=   \log_2({1+\text{SINR}_{U,n}}). \label{rate_ob2}
\end{align}
\endgroup

\subsection{Problem Formulation}
	We aim for maximizing the DL rate while guaranteeing UL rate requirement, which is represented by
\begingroup
\allowdisplaybreaks
\begin{subequations} \label{P1}
\begin{align}
\mathop{\max}_{{\boldsymbol{\Theta}},{\rm\mathbf{W}}} \ & \sum_{m\in \mathcal{M}} R_{D,m}({\boldsymbol{\Theta}}, {\rm\mathbf{W}}) \label{P1:obj_func}\\
\text{s.t. } 
& \sum_{n \in \mathcal{N}} R_{U,n}({\boldsymbol{\Theta}}, {\rm\mathbf{W}}) \geq t_{th,U}, \label{P1:QoS_constraint}\\
& \left| e^{j\theta_k} \right|^2 = 1, \quad \forall k\in \mathcal{K}, \label{P1:phase_constraint}\\
&\sum_{m\in \mathcal{M}} {\rm\mathbf{w}}_m^{H} {\rm\mathbf{w}}_m \leq P_{\text{max}}.\label{P1:power_constraint}
\end{align}
\end{subequations}
\endgroup
Constraint \eqref{P1:QoS_constraint} guarantees the UL rate above a threshold $t_{th,U}$. Constraint \eqref{P1:phase_constraint} stands for RIS configuration, whilst \eqref{P1:power_constraint} presents the maximum BS transmit power $P_{\text{max}}$. $H$ denotes the Hermitian operation of a matrix. We can know that problem $\eqref{P1}$ is non-convex and non-linear owing to fractional functions, logarithmic expressions and the equality constraint of RIS, which will be respectively solved below.

\section{Proposed FRIS Scheme}

\subsection{Logarithmic and Fractional Transformation}\label{subsection:1}

We can observe from $\eqref{rate_ob1}$ and $\eqref{rate_ob2}$ that the objective function performs summation of logarithmic functions with each having fractional expression of SINR. To tackle the issues, we employ Lagrangian dual transform \cite{Lagran_Dual_Transform} and the Dinkelbach method \cite{Dinkelbach}. Firstly, we transform the constraint of $\eqref{P1:QoS_constraint}$ due to multiple summation operations. Based on Jensen's inequality, we have a lower  bound of total UL rate bound as $\sum_{n\in \mathcal{N}} R_{U,n}({\boldsymbol{\Theta}}, {\rm\mathbf{W}}) \geq R_{U}({\boldsymbol{\Theta}}, {\rm\mathbf{W}}) =  \log_2({1+\text{SINR}_{U}})$, where $\text{SINR}_{U}
    =\frac{\lVert ({\rm\mathbf{U}}+{\rm\mathbf{U}}_t {\boldsymbol{\Theta}}^{(U)})\overline{{\rm\mathbf{x}}}_U \lVert^2}{\lVert ({\rm\mathbf{S}}+{\rm\mathbf{S}}_t {\boldsymbol{\Theta}}^{(U)}){\rm\mathbf{W}}\overline{{\rm\mathbf{x}}}_D\lVert^2 + \sigma^2}
    \triangleq \frac{U_{S}({\boldsymbol{\Theta}})}{U_{I}({\boldsymbol{\Theta}}, {\rm\mathbf{W}}) + \sigma^2 }$ is the equivalent expression of total SINR of UL UEs.
Therefore, the transformed problem of $\eqref{P1}$ is given by
\begin{subequations} \label{P1_1}
\begin{align}
\mathop{\max}_{{\boldsymbol{\Theta}},{\rm\mathbf{W}}} \ & \sum_{m\in \mathcal{M}} R_{D,m}({\boldsymbol{\Theta}}, {\rm\mathbf{W}}) \label{P1_1:obj_func0}\\
\text{s.t. } 
& R_{U}({\boldsymbol{\Theta}}, {\rm\mathbf{W}}) \geq t_{th,U}, \label{P1_1:QoS_constraint} \quad \eqref{P1:phase_constraint}, \eqref{P1:power_constraint}.
\end{align}
\end{subequations}
In the following proposition, we design to remove the logarithmic function in $\eqref{P1_1:obj_func0}$.

\begin{prop} \label{prop0}
    The sum-of-logarithm problem $\eqref{P1_1}$ is equivalent to solving
\begin{subequations} \label{P1_3}
    \begin{align}
    & \mathop{\max}_{{\boldsymbol{\Theta}},{\rm\mathbf{W}}} \  \overline{F}_{D}({\boldsymbol{\Theta}}, {\rm\mathbf{W}}) 
 \notag   \\
    &= \sum_{m\in\mathcal{M}} \frac{(1+r_{m}^{*})D_{S,m}({\boldsymbol{\Theta}}, {\rm\mathbf{W}})}{D_{S,m}({\boldsymbol{\Theta}}, {\rm\mathbf{W}}) \!+\! D_{I,m}({\boldsymbol{\Theta}}, {\rm\mathbf{W}}) \!+\! D_{C,m}({\boldsymbol{\Theta}}) \!+\! \sigma^2} \label{P1_1:obj_func}\\
    & \qquad \text{\rm s.t. } 
    \eqref{P1_1:QoS_constraint},
    \end{align}
\end{subequations}
where ${r}_m^{*} = \frac{D_{S,m}({\boldsymbol{\Theta}}, {\rm\mathbf{W}})}{ D_{I,m}({\boldsymbol{\Theta}}, {\rm\mathbf{W}}) + D_{C,m}({\boldsymbol{\Theta}}) + \sigma^2}$ is auxiliary term corresponding to the optimal DL SINR.
\end{prop}
\begin{proof} 
Based on Lagrangian dual transform, we have weighted sum-of-logarithms given by
\begingroup
\allowdisplaybreaks
\begin{align}
    & F_{D}({\boldsymbol{\Theta}}, {\rm\mathbf{W}}, \boldsymbol{r}) = \sum_{m\in \mathcal{M}} \log_{2}(1+r_m) - \sum_{m\in \mathcal{M}} r_m  \notag\\
    & + \sum_{m\in \mathcal{M}} \frac{(1+r_m)D_{s,m}({\boldsymbol{\Theta}}, {\rm\mathbf{W}})}{D_{S,m}({\boldsymbol{\Theta}}, {\rm\mathbf{W}}) + D_{I,m}({\boldsymbol{\Theta}}, {\rm\mathbf{W}}) + D_{C,m}({\boldsymbol{\Theta}}) + \sigma^2}
\end{align}
\endgroup
where $\boldsymbol{r} = [r_1, r_2,...,r_M]$ indicates the auxiliary variable of DL SINR. We can observe that $F_{D}({\boldsymbol{\Theta}}, {\rm\mathbf{W}}, \boldsymbol{r})$ is convex and differentiable function with respect to (w.r.t.) $r^{*}_{m}$, with fixed BS/RIS parameters $\{ \boldsymbol{\Theta}, \mathbf{W} \}$. Therefore, the optimum of $r^{*}_{m}$ can be obtained by taking $\frac{\partial F_{D}({\boldsymbol{\Theta}}, {\rm\mathbf{W}}, \boldsymbol{r})}{\partial r_m}=0$. Substituting $r^{*}_m$ in $F_{D}({\boldsymbol{\Theta}}, {\rm\mathbf{W}}, \boldsymbol{r})$ yields $\eqref{P1_1:obj_func}$. This completes the proof.
\end{proof}

Note that $r_m^*$ can be set as the optimum of DL SINR acquired at previous iteration. We employ the Dinkelbach method in $\eqref{P1_1:obj_func}$, which converts the fractional expression into an affine function as
\begingroup
\allowdisplaybreaks
\begin{align}
	\widetilde{F}_{D}({\boldsymbol{\Theta}}, {\rm\mathbf{W}}) &= \sum_{m\in \mathcal{M}} ( 1 + r_{m}^{*} )D_{S,m}({\boldsymbol{\Theta}}, {\rm\mathbf{W}}) 
	- t_m^{*}\Big( D_{S,m}({\boldsymbol{\Theta}}, {\rm\mathbf{W}}) \notag\\
	& + D_{I,m}({\boldsymbol{\Theta}}, {\rm\mathbf{W}}) + D_{C,m}({\boldsymbol{\Theta}}) + \sigma^2 \Big),
\end{align}
\endgroup
where the optimal solution of the Lagrangian dual in the $i$-th iteration is defined as the optimum obtained in the previous iteration denoted by
\begin{equation}
    \begin{aligned} \label{eqn: obj_ratio_D}
    t_m^{*} = \frac{(1+r_{m}^{*})D_{s,m}({\boldsymbol{\Theta}}^{*(i-1)}, {\rm\mathbf{W}}^{*(i-1)})}{A_1 + A_2 + A_3 + \sigma^2},
    \end{aligned}
\end{equation}
where $A_1 \!=\! D_{S,m}({\boldsymbol{\Theta}}^{*(i-1)}, {\rm\mathbf{W}}^{*(i-1)})$, 
$A_2 \!=\! D_{I,m}({\boldsymbol{\Theta}}^{*(i-1)}, {\rm\mathbf{W}}^{*(i-1)})$ and 
$A_3 = D_{C,m}({\boldsymbol{\Theta}}^{*(i-1)})$. Moreover, UL QoS in $\eqref{P1_1:QoS_constraint}$ can be rewritten as 
\begin{equation}
    \begin{aligned} \label{eqn: UL_QoS}
    \text{SINR}_U \geq \overline{t}_{th,U}
    \ 
    \Rightarrow 
    \
    U_S(\boldsymbol{\Theta}) \geq \overline{t}_{th,U}(U_I(\boldsymbol{\Theta}, {\rm\mathbf{W}})+ \sigma^2),
    \end{aligned}
\end{equation}
where $\overline{t}_{th,U} = 2^{t_{th,U}}-1$. In the last inequality of $\eqref{eqn: UL_QoS}$, the denominator of SINR $\eqref{eqn: obj_ratio_D}$ is moved to the right-hand side. Therefore, the problem $\eqref{P1_3}$ in the $i$-th iteration can be acquired as
\begin{align} \label{P1_4}
    \mathop{\max}_{{\boldsymbol{\Theta}},{\rm\mathbf{W}}} \ &  \widetilde{F}_{D}({\boldsymbol{\Theta}}, {\rm\mathbf{W}})
    \quad \text{ s.t. } 
    \eqref{P1:phase_constraint}, \eqref{P1:power_constraint}, \eqref{eqn: UL_QoS}.
\end{align}
However, we can infer that problem $\eqref{P1_4}$ still performs non-convexity and non-linearity due to coupled variables of $\{ {\boldsymbol{\Theta}},{\rm\mathbf{W}} \}$. Therefore, we decouple problem $\eqref{P1_4}$ into two subproblems by AO algorithm, including active BS beamforming and passive RIS phase shifts, which are elaborated as follows.

\subsection{Alternative Optimization}

\subsubsection{Optimization of Active BS Transmit Beamforming}

While the RIS phase matrix $\boldsymbol{\Theta^{*}}$ is fixed, problem \eqref{P1_4} is reduced to 
\begin{subequations} \label{subP2}
    \begin{align}
    \mathop{\max}_{{\rm\mathbf{W}}} \ &  \widetilde{F}_{D}^{w}({\rm\mathbf{W}}) 
    \!=\! \sum_{m\in \mathcal{M}} (1 \!+\! r_m^*) {D}_{S,m}({\rm\mathbf{W}}) \!-\! t_m^{*}({D}_{t,m}({\rm\mathbf{W}}) \!+\! n_m) \label{subP2:obj_func}\\
    \text{s.t. } 
    & \eqref{P1:power_constraint}, \quad 
    {U}_I({\rm\mathbf{W}}) \leq \xi_U, \label{subP2: UL_constraint}
    \end{align}
\end{subequations}
where ${D}_{S,m}({\rm\mathbf{W}}) \!=\! \lVert ({\rm\mathbf{D}}_m + {\rm\mathbf{D}}_{t_{m}}{\boldsymbol{\Theta}}^{{*(D)}}){\rm\mathbf{w}}_m {\overline{x}_{D,m}} \lVert^2
    \triangleq \lVert {\rm\mathbf{D}}_{w,m} {\rm\mathbf{w}}_m \lVert^2$
    , ${D}_{t,m}({\rm\mathbf{W}}) 
	\!=\! \sum_{m'\in\mathcal{M}}
	\lVert ({\rm\mathbf{D}}_m + {\rm\mathbf{D}}_{t_{m}}{\boldsymbol{\Theta}}^{{*(D)}})
	\\ 
	{\rm\mathbf{w}}_{m'} {\overline{x}}_{D,m^{'}} \lVert^2
    \triangleq \sum_{m'\in\mathcal{M}} \lVert {\rm\mathbf{D}}_{w,m} {\rm\mathbf{w}}_{m'} \lVert^2$
    , and $n_m \!=\! {D}_{C,m}({\boldsymbol{\Theta}}^{*}) + \sigma^2$. By moving beamforming parameters to the left-hand side of inequality in $\eqref{eqn: UL_QoS}$, we have 
    ${U}_I({\rm\mathbf{W}}) 
    = \lVert ({\rm\mathbf{S}}+{\rm\mathbf{S}}_t {\boldsymbol{\Theta}}^{{*(D)}}) {\rm\mathbf{W}} \overline{\rm\mathbf{x}}_D \lVert^2 
    \triangleq \lVert {\rm\mathbf{S}}_{w} \sum_{m\in \mathcal{M}} {\rm\mathbf{w}}_m \lVert^2$
    and $\xi_U = \frac{{U}_S({\boldsymbol{\Theta}}^*) - \overline{t}_{th,U} \sigma^2}{\overline{t}_{th,U}}$. Note that these terms are related to only beamforming parameter $\mathbf{W}$ with the fixed optimal RIS solution of $\boldsymbol{\Theta}^*$.

Now, the DL/UL terms related to beamforming in problem $\eqref{subP2}$ can be further expressed as quadratic form. The DL objective in $\eqref{subP2:obj_func}$ is expressed as
\begin{align}\label{eqn: obj_w_t}
    \widetilde{F}_{D}({\rm\mathbf{W}}) \!=\! \smashoperator[r]{\sum_{m\in\mathcal{M}}} {\rm\mathbf{w}}_m^{H} {\boldsymbol{\Omega}}_{w,m} {\rm\mathbf{w}}_m
    \!-\! t_m^{*} \left( \smashoperator[r]{\sum_{m'\in\mathcal{M}}} {\rm\mathbf{w}}_{m'}^{H} {\boldsymbol{\Omega}}_{w,m} {\rm\mathbf{w}}_{m'} \!+\! n_m \right)
\end{align}
where ${\boldsymbol{\Omega}}_{w,m} = (1+r_m^*){\rm\mathbf{D}}_{w,m}^{H} {\rm\mathbf{D}}_{w,m}$.
As for UL QoS in $\eqref{subP2: UL_constraint}$, we have
\begin{align} \label{eqn: ULcons_w_t}
    \quad {U}_I({\rm\mathbf{W}}) 
    &= \sum_{m\in\mathcal{M}} {\rm\mathbf{w}}_m^{H} {\boldsymbol{\Omega}}_{w,U} \sum_{m\in\mathcal{M}} {\rm\mathbf{w}}_m
    = \sum_{m\in\mathcal{M}} {\rm\mathbf{w}}_m^{H} {\boldsymbol{\Omega}}_{w,U} {\rm\mathbf{w}}_m \notag \\
    & \qquad + 2\sum_{m=1}^{M-1}  \sum_{m'=m+1}^{M} {\rm\mathbf{w}}_m^{H} {\boldsymbol{\Omega}}_{w,U} {\rm\mathbf{w}}_{m'} \leq \xi_U,
\end{align}
where ${\boldsymbol{\Omega}}_{w,U} = {\rm\mathbf{S}}_{w}^{H} {\rm\mathbf{S}}_{w}$. However, the objective $\eqref{eqn: obj_w_t}$ performs a convex-concave function, whilst beamforming variables are coupled in $\eqref{eqn: ULcons_w_t}$. Therefore, SCA is employed to transform the non-convex term to convex one. Based on the first-order Taylor approximation, we can acquire the alternative lower bounds as
\begingroup
\allowdisplaybreaks
\begin{align}
	&{\rm\mathbf{w}}_m^{H}  {\boldsymbol{\Omega}}_{w,m} {\rm\mathbf{w}}_m 
	\geq 
	2\text{Re}\{ ({\rm\mathbf{w}}_m^{(p)}) ^{H} {\boldsymbol{\Omega}}_{w,m} {\rm\mathbf{w}}_m\} - ({\rm\mathbf{w}}_m^{(p)}) ^{H} {\boldsymbol{\Omega}}_{w,m} {\rm\mathbf{w}}_m^{(p)} \label{Taylor_1} \\
	&\smashoperator[r]{\sum_{m=1}^{M-1}} \smashoperator[r]{\sum_{m'=m+1}^{M}} {\rm\mathbf{w}}_m^{H} {\boldsymbol{\Omega}}_{w,U} {\rm\mathbf{w}}_{m'}
	\geq 
	\smashoperator[r]{\sum_{m=1}^{M-1}} \smashoperator[r]{\sum_{m'=m+1}^{M}} 2\text{Re}\{ ({\rm\mathbf{w}}_m^{(p)}) ^{H} {\boldsymbol{\Omega}}_{w,U} {\rm\mathbf{w}}_{m'} \} \notag \\
	& \qquad\qquad\qquad\qquad\qquad\qquad\qquad - ({\rm\mathbf{w}}_m^{(p)}) ^{H} {\boldsymbol{\Omega}}_{w,U}{\rm\mathbf{w}}_{m'}^{(p)} \label{Taylor_2}.
\end{align}
\endgroup
We define $\mathbf{w}_m^{(p)}$ as a constant solution of active beamforming obtained at iteration $p$. Based on $\eqref{eqn: obj_w_t}$ and $\eqref{Taylor_1}$, the objective of $\eqref{subP2:obj_func}$ is transformed into 
\begin{align}\label{eqn: obj_w_Taylor}
    \widetilde{F}_{D,t}( {\rm\mathbf{W}})
    &= \sum_{m\in \mathcal{M}} - t_m^{*} \left( \sum_{m'\in \mathcal{M}} {\rm\mathbf{w}}_{m'}^{H}  {\boldsymbol{\Omega}}_{w,m} {\rm\mathbf{w}}_{m'} \right)  \notag\\
   & \qquad + 2\text{Re}\{({\rm\mathbf{w}}_m^{(p)}) ^{H} {\boldsymbol{\Omega}}_{w,m} {\rm\mathbf{w}}_m\} + c_{w,m}, 
\end{align}
where $c_{w,m} = - ({\rm\mathbf{w}}_m^{(p)}) ^{H} {\boldsymbol{\Omega}}_{w,m} {\rm\mathbf{w}}_m^{(p)} - t_m^{*} {n}_m$ is a constant. According to $\eqref{eqn: ULcons_w_t}$ and $\eqref{Taylor_2}$, the UL QoS constraint of $\eqref{subP2: UL_constraint}$ is converted into
\begin{align}\label{eqn: ULcons_w_Taylor}
    & {U}_I({\rm\mathbf{W}}) \leq \xi_U \Rightarrow \left(\sum_{m\in \mathcal{M}} {\rm\mathbf{w}}_m^{H} {\boldsymbol{\Omega}}_{w,U} {\rm\mathbf{w}}_m\right) \notag \\
    & + 2\left( \sum_{m=1}^{M-1} \smashoperator[r]{\sum_{m'=m+1}^{M}} 2\text{Re}\{ ({\rm\mathbf{w}}_m^{(p)}) ^{H} {\boldsymbol{\Omega}}_{w,U} {\rm\mathbf{w}}_{m'} \}\right) \!+\! c_{w,U} \leq 0,
\end{align}
where $c_{w,U} = -2\sum_{m=1}^{M-1}\sum_{m'=m+1}^{M} ({\rm\mathbf{w}}_m^{(p)}) ^{H} {\boldsymbol{\Omega}}_{w,U}{\rm\mathbf{w}}_{m'}^{(p)} - \xi_U $ is a constant. Finally, we have the transformed problem w.r.t. $\mathbf{W}$ as
\begin{align} \label{subP2.2}
    \mathop{\max}_{{\rm\mathbf{W}}} \ &  \widetilde{F}_{D,t} ( {\rm\mathbf{W}}) 
    \quad \text{s.t. } 
    \eqref{P1:power_constraint}, \eqref{eqn: ULcons_w_Taylor},
    \end{align}
which is convex and can be solved to obtain the optimal active BS beamforming solution.

\subsubsection{Optimization of Passive RIS Phase-Shifts}

After obtaining the optimal active beamforming solution of $\mathbf{W}^*$, we now consider the subproblem w.r.t. passive beamforming of RIS $\boldsymbol{\Theta}$. Firstly, we employ a vector form of RIS phase shift of $\boldsymbol{\theta} = [e^{j\theta_1},...,e^{j\theta_K}]^{T} \in \mathbb{C}^K$ instead of matrix form of $\boldsymbol{\Theta}$. Similar to $\eqref{subP2}$, the problem $\eqref{P1_4}$ can be reformulated w.r.t. $\boldsymbol{\theta}$ with fixed $\mathbf{W}^*$ as
\begingroup
\allowdisplaybreaks
\begin{subequations} \label{subP1}
    \begin{align}
    \mathop{\max}_{\boldsymbol{\theta}} \ &  \widetilde{F}_{D}^{\theta}({\boldsymbol{\theta}})
    \!=\! \smashoperator[r]{\sum_{m\in \mathcal{M}}} {D}_{S,m}({\boldsymbol{\theta}}) - t_m^{*}({D}_{t,m}({\boldsymbol{\theta}})+ {D}_{C,m}({\boldsymbol{\theta}}) + \sigma^2) \label{subP1:obj_func}\\
    \text{s.t. } 
    & \eqref{P1:phase_constraint}, 
    \quad
    {U}_S(\boldsymbol{\theta}) - \overline{t}_{th,U}\left({U}_I({\boldsymbol{\theta}}) + \sigma^{2} \right) \geq 0, \label{subP1.2: UL constraint}
    \end{align}
\end{subequations}
\endgroup
where pertinent notations w.r.t. $\boldsymbol{\theta}$ are defined as follows:
	${D}_{S,m}({\boldsymbol{\theta}})
	\!=\! (1+r_m^*)\lVert {\rm\mathbf{D}}_m {\rm\mathbf{w}}_m^{*} + {\rm\mathbf{D}}_{t_{m}}{\rm\mathbf{w}}_m^{*(D)}{\boldsymbol{\theta}} \lVert^2$,
	${\rm\mathbf{w}}_{m}^{*(D)} \!=\! {\rm diag}({\rm\mathbf{w}}_m^*,{\rm\mathbf{w}}_m^*,...,{\rm\mathbf{w}}_m^*) \in \mathbb{C}^{N \cdot K \times K}$,
	${D}_{t,m}({\boldsymbol{\theta}})
	\!=\! \sum_{m'\in \mathcal{M}} \lVert {\rm\mathbf{D}}_m {\rm\mathbf{w}}_{m'}^{*} + {\rm\mathbf{D}}_{t_{m}}{\rm\mathbf{w}}_{m'}^{*(D)}{\boldsymbol{\theta}} \lVert^2$,
	$\overline{\rm\mathbf{x}}_U^{(U)} \!=\! {\rm diag}(\overline{\rm\mathbf{x}}_U,\overline{\rm\mathbf{x}}_U,...,\overline{\rm\mathbf{x}}_U)\in \mathbb{C}^{N K \times K}$,
	${D}_{C,m}({\boldsymbol{\theta}})
	\!=\! \lVert {\rm\mathbf{V}}_m \overline{{\rm\mathbf{x}}}_U + {\rm\mathbf{C}}_{t_{m}}\overline{\rm\mathbf{x}}_U^{(U)}{\boldsymbol{\theta}} \lVert^2,
	\quad {\rm\mathbf{d}}_{t}^{(D)} = {\rm diag}({\rm\mathbf{d}}_t,...,{\rm\mathbf{d}}_t) \in \mathbb{C}^{N_t K \times K}$, 
	${\rm\mathbf{d}}_t \!=\! \sum_{m\in \mathcal{M}} {\rm\mathbf{w}}_m^{*} \overline{x}_{D,m} \in \mathbb{C} ^{N_t}$,
	${U}_{S}({\boldsymbol{\theta}})
	\!=\! \lVert {\rm\mathbf{U}}\overline{{\rm\mathbf{x}}}_U+{\rm\mathbf{U}}_t \overline{{\rm\mathbf{x}}}_U ^{(U)} {\boldsymbol{\theta}} \lVert^2$,
	and
	${U}_I({\boldsymbol{\theta}}) 
	\!=\! \lVert {\rm\mathbf{S}}{\rm\mathbf{W}}^{*}\overline{{\rm\mathbf{x}}}_D+{\rm\mathbf{S}}_t {\rm\mathbf{d}}_t^{(D)} {\boldsymbol{\theta}}\lVert^2,$
Therefore, we can have equivalent expressions by moving out the parameter $\boldsymbol{\theta}$ originally within two matrices. Accordinly, by rearranging the above formulas, the objective $\eqref{subP1:obj_func}$ can be rewritten in an quadratic form as 
\begin{equation}
    \begin{aligned} \label{eqn: obj_tran}
    \widetilde{F}_{D}({\boldsymbol{\theta}}) = \sum_{m\in \mathcal{M}} {\boldsymbol{\theta}}^{H} \boldsymbol{\Omega}_{\theta,m}{\boldsymbol{\theta}} - {\boldsymbol{\theta}}^{H} \boldsymbol{\Psi}_{\theta,m}{\boldsymbol{\theta}} + 2\text{Re}\{{\boldsymbol{\zeta}}_{\theta,m}{\boldsymbol{\theta}}\} + c_{\theta,m},
    \end{aligned}
\end{equation}
where 
	$\boldsymbol{\Omega}_{\theta,m} = (1+r_m^*)({\rm\mathbf{w}}_m^{*(D)})^{H} {\rm\mathbf{D}}_{t_m}^{H} {\rm\mathbf{D}}_{t_m} {\rm\mathbf{w}}_m^{*(D)}$,
	$\boldsymbol{\Psi}_{\theta,m} = t_m^* {\rm\mathbf{D}}_{t_m} \sum_{m'\in \mathcal{M}} \left( {\rm\mathbf{w}}_{m'}^{*(D)} ({\rm\mathbf{w}}_{m'}^{*(D)})^{H} \right) {\rm\mathbf{D}}_{t_m}^{H}$,
	${\boldsymbol{\zeta}}_{\theta,m} =  (1+r_m^*) ({\rm\mathbf{w}}_m^*)^{H} {\rm\mathbf{D}}_m^{H}  {\rm\mathbf{D}}_{t_m} {\rm\mathbf{w}}_m^{*(D)} 
    - t_m^* \left( \sum_{m'\in \mathcal{M}} ({\rm\mathbf{w}}_{m'}^*)^{H} {\rm\mathbf{D}}_m^{H}  {\rm\mathbf{D}}_{t_m} {\rm\mathbf{w}}_{m'}^{*(D)} +  \overline{\rm\mathbf{x}}_{U}^{H} {\rm\mathbf{V}}_m^{H} {\rm\mathbf{C}}_{t_m} \overline{{\rm\mathbf{x}}}_U^{(U)} \right)$,
    and
    $c_{\theta,m} = (1+r_m^*)({{\rm\mathbf{w}}_m^{*}})^{H} {\rm\mathbf{D}}_m^{H} {\rm\mathbf{D}}_m {\rm\mathbf{w}}_m^* 
    - t_m^* \left( \sum_{m'\in \mathcal{M}} \left({{\rm\mathbf{w}}_{m'}^*}\right)^{H} {\rm\mathbf{D}}_{m}^{H} {\rm\mathbf{D}}_m {\rm\mathbf{w}}_{m'}^* + \overline{\rm\mathbf{x}}_U^{H} {\rm\mathbf{V}}_m^{H} {\rm\mathbf{V}}_m \overline{\rm\mathbf{x}}_U + \sigma^2 \right)$. Moreover, UL QoS constraint in \eqref{subP1.2: UL constraint} can be rewritten in a quadratic form given by
 \begin{equation}
    \begin{aligned} \label{eqn: cons_tran} {\boldsymbol{\theta}}^{H} \boldsymbol{\Omega}_{\theta,U} {\boldsymbol{\theta}} - {\boldsymbol{\theta}}^{H} \boldsymbol{\Psi}_{\theta,U} {\boldsymbol{\theta}} + 2\text{Re}\{{\boldsymbol{\zeta}}_{\theta,U}{\boldsymbol{\theta}}\} + c_{\theta,U} \geq 0,
    \end{aligned}
\end{equation}
where 
	$\boldsymbol{\Omega}_{\theta,U} = (\overline{\rm\mathbf{x}}_U^{(U)})^{H} {\rm\mathbf{U}}_t^{H} {\rm\mathbf{U}}_t \overline{\rm\mathbf{x}}_U^{(U)}$,
	${\boldsymbol{\zeta}}_{\theta,U} = \overline{\rm\mathbf{x}}_U^{H} {\rm\mathbf{U}}^{H} {\rm\mathbf{U}}_t \overline{\rm\mathbf{x}}_U^{(U)}- \overline{t}_{th,U} \overline{\rm\mathbf{x}}_D^{H} ({{\rm\mathbf{W}}^*})^{H} {\rm\mathbf{S}}^{H} {\rm\mathbf{S}}_t {\rm\mathbf{d}}_t^{(D)}$,
	$\boldsymbol{\Psi}_{\theta,U} = \overline{t}_{th,U} ({\rm\mathbf{d}}_t^{(D)})^{H} {\rm\mathbf{S}}_t^{H} {\rm\mathbf{S}}_t {\rm\mathbf{d}}_t^{(D)}$,
	and
	$c_{\theta,U} = \overline{\rm\mathbf{x}}_U^{H} {\rm\mathbf{U}}^{H} {\rm\mathbf{U}} \overline{\rm\mathbf{x}}_U -  \overline{t}_{th,U} \left( \overline{\rm\mathbf{x}}_D^{H} ({{\rm\mathbf{W}}^*})^{H} {\rm\mathbf{S}}^{H} {\rm\mathbf{S}} {{\rm\mathbf{W}}^*}\overline{\rm\mathbf{x}}_D + \sigma^2 \right)$.
However, we can observe that $\eqref{eqn: obj_tran}$ and $\eqref{eqn: cons_tran}$ both perform convex-concave functions. Therefore, we employ SCA with the first-order Taylor approximation to transform a non-convex function to affine one, where the inequalities are given by ${\boldsymbol{\theta}}^{H} \boldsymbol{\Omega}_{\theta,m} {\boldsymbol{\theta}} \geq 2\text{Re}\{ ({\boldsymbol{\theta}}^{(q)})^{H} {\boldsymbol{\Omega}}_{\theta,m} {\boldsymbol{\theta}}\} - ({\boldsymbol{\theta}}^{(q)})^{H} {\boldsymbol{\Omega}}_{\theta,m} {\boldsymbol{\theta}}^{(q)}$ in \eqref{eqn: obj_tran}, and ${\boldsymbol{\theta}}^{H} \boldsymbol{\Omega}_{\theta,U} {\boldsymbol{\theta}} \geq 2\text{Re}\{ ({\boldsymbol{\theta}}^{(q)})^{H} \boldsymbol{\Omega}_{\theta,U} {\boldsymbol{\theta}}\} - ({\boldsymbol{\theta}}^{(q)})^{H} \boldsymbol{\Omega}_{\theta,U} {\boldsymbol{\theta}}^{(q)}$ in $\eqref{eqn: cons_tran}$, where $\boldsymbol{\theta}^{(q)}$ indicates a constant of RIS solution obtained at the previous iteration $q$. Accordingly, we can attain an SCA-based objective in $\eqref{eqn: obj_tran}$ as $\widetilde{F}_{D,t}({\boldsymbol{\theta}})
    = \sum_{m\in \mathcal{M}} {\boldsymbol{\theta}}^{H} \boldsymbol{\Psi}_{\theta,m}{\boldsymbol{\theta}} + 2\text{Re}\{\overline{\boldsymbol{\zeta}}_{\theta,m}{\boldsymbol{\theta}}\} + \overline{c}_{\theta,m}$,
where $\overline{\boldsymbol{\zeta}}_{\theta,m} = {\boldsymbol{\zeta}}_{\theta,m} + ({\boldsymbol{\theta}}^{(q)})^{H} {\boldsymbol{\Omega}}_{\theta,m} $ 
and 
$\overline{c}_{\theta,m} = c_{\theta,m} -({\boldsymbol{\theta}}^{(q)})^{H} {\boldsymbol{\Omega}}_{\theta,m} {\boldsymbol{\theta}}^{(q)}$. Similarly, with SCA, UL QoS constraint of $\eqref{eqn: cons_tran}$ can be expressed as
\begin{equation} \label{SCA_ULQoS}
    \begin{aligned}
     - {\boldsymbol{\theta}}^{H} \boldsymbol{\Psi}_{\theta,U} {\boldsymbol{\theta}} + 2\text{Re}\{\overline{\boldsymbol{\zeta}}_{\theta,U}{\boldsymbol{\theta}}\} + \overline{c}_{\theta,U} \geq 0,
    \end{aligned}
\end{equation}
where $\overline{\boldsymbol{\zeta}}_{\theta,U} = ({\boldsymbol{\theta}}^{(q)})^{H} \boldsymbol{\Omega}_{\theta,U} + {\boldsymbol{\zeta}}_{\theta,U}$ and $\overline{c}_{\theta,U} = c_{\theta,U} - ({\boldsymbol{\theta}}^{(q)})^{H} \boldsymbol{\Omega}_{\theta,U} {\boldsymbol{\theta}}^{(q)}$. Accordingly, the problem $\eqref{subP1}$ is transformed into
\begin{subequations} \label{subP1.3}
    \begin{align}
    \mathop{\max}_{\boldsymbol{\theta}} \ &  \widetilde{F}_{D,t} ({\boldsymbol{\theta}})
    \label{subP1.3:obj_func}
    \quad \text{ s.t. } 
    \eqref{P1:phase_constraint}, \eqref{SCA_ULQoS}.
    \end{align}
\end{subequations}
Nevertheless, the equality of RIS phase constraint in $\eqref{P1:phase_constraint}$ is non-convex. To address this problem, we adopt a convex-concave programming (CCP) procedure, i.e., $\lvert e^{j\theta_k} \lvert^2 = 1$ can be resolved by employing two equivalent inequalities of $\lvert e^{j\theta_k} \lvert^2 \leq 1 $ and $\lvert e^{j\theta_k} \lvert^2 \geq 1 $. The first inequality is convex due to its circular shape, whilst the second inequality is non-convex. Adopting SCA again for $\lvert e^{j\theta_k} \lvert^2 \geq 1$ yields $2\text{Re}\{ (e^{-j\theta_k})^{(q)} e^{j\theta_k}\} - \lvert (e^{j\theta_k})^{(q)} \lvert^2 \geq 1 \Rightarrow \text{Re}\{ (e^{-j\theta_k})^{(q)} e^{j\theta_k}\} \geq 1$ since $\lvert (e^{j\theta_k})^{(q)} \lvert^2=1$, where $(e^{-j\theta_k})^{(q)}$ indicates the phase solution obtained as a constant at iteration $q$. However, directly applying these tight constraints may lead to non-convergence or an infeasible problem during iterative optimization. Therefore, PCCP is employed to allow optimization from loose to tight constraints, asymptotically approaching the original equality constraint. The optimization problem w.r.t. RIS configuration of $\boldsymbol{\theta}$ can be expressed as
\begingroup
\allowdisplaybreaks
\begin{subequations} \label{subP1.4}
    \begin{align}
    \mathop{\max}_{{\boldsymbol{\theta}},{\rm\mathbf{a}},{\rm\mathbf{b}}} \ &  \widetilde{F}_{D,t}({\boldsymbol{\theta}})- \lambda^{(q)} \sum_{k\in \mathcal{K}} (a_k+b_k)
    \label{subP1.3:obj_func2}\\
    \text{s.t. } 
    & \eqref{SCA_ULQoS}, \quad \left\lvert e^{j\theta_k} \right\lvert^2 \leq 1 + a_k, \quad \forall k\in \mathcal{K},\\
    & \text{Re}\{ (e^{-j\theta_k})^{(q)} e^{j\theta_k}\} \geq 1-b_k, \quad \forall k\in \mathcal{K},
    \end{align}
\end{subequations}
\endgroup
where ${\rm\mathbf{a}} = [a_1,...,a_K]^{T}$ and ${\rm\mathbf{b}} = [b_1,...,b_K]^{T}$ denote the slake variables imposed on the associated constraints of RIS phase shifts. To this end, we have a convex problem $\eqref{subP1.4}$, which can be readily resolved by arbitrary convex optimization tools. The concrete step of the proposed FRIS scheme is described in Algorithm \ref{FRIS}. The whole procedure will be iteratively executed until their respective convergence given by $| \widetilde{F}_{D,t}^{(p)} ( {\rm\mathbf{W}}) - \widetilde{F}_{D,t}^{(p-1)} ( {\rm\mathbf{W}}) | \leq \varrho_w$, $| \widetilde{F}_{D,t}^{(q)} ( {\boldsymbol{\theta}}) - \widetilde{F}_{D,t}^{(q-1)} ( {\boldsymbol{\theta}}) | \leq \varrho_{\theta}$, and $| \widetilde{F}_{D}^{(i)}({\boldsymbol{\Theta}}, {\rm\mathbf{W}}) - \widetilde{F}_{D}^{(i-1)}({\boldsymbol{\Theta}}, {\rm\mathbf{W}})| \leq \varrho$, where $ \varrho_w$, $\varrho_{\theta}$, and $\varrho$ indicate the respective convergence bounds. To elaborate a little further, we select the historical optimal solution if the iteration exceeds their respective upper bounds, i.e., $p > T_w$, $q > T_{\theta}$, and $i > T$.

\begin{algorithm}[t]
\caption{\footnotesize Proposed FRIS Scheme}
\SetAlgoLined
\scriptsize
\DontPrintSemicolon
\label{FRIS}
\begin{algorithmic}[1]
\STATE Set iteration of outer loop $i=1$
\REPEAT [Outer Loop for Lagrangian/Dinkelbach]
\STATE Compute parameters of $r_m^{(i)}$, $t_m^{(i)}$ based on ${\rm\mathbf{W}}^{(i-1)}$, ${\boldsymbol{\theta}}^{(i-1)}$ 
	\STATE Set inner iteration $p=1$ and $\mathbf{W}_{t}^{(0)} = {\rm\mathbf{W}}^{(i-1)}$
	\REPEAT [Inner Loop for BS Beamforming]	
		\STATE Solve problem $\eqref{subP2.2}$ for ${\rm\mathbf{W}}_{t}^{(p)}$ with fixed $\boldsymbol{\theta}^{(i-1)}$; Update $p \leftarrow p+1$
	\UNTIL Convergence of $\mathbf{W}_{t}$
	\STATE Obtain $\mathbf{W}_{t}^{*} = \mathbf{W}_{t}^{(p-1)}$
	\STATE Set inner iteration $q=1$ and $\boldsymbol{\theta}_{t}^{(0)} = \boldsymbol{\theta}^{(i-1)}$, $\kappa\geq 0$, $\lambda^{(0)}=1$ and $\lambda_{\text{max}}$
	\REPEAT [Inner Loop for RIS Beamforming]
		\STATE Solve problem $\eqref{subP1.4}$ for $\boldsymbol{\theta}_{t}^{(q)}$ with fixed ${\rm\mathbf{W}}_{t}^{*}$
		\STATE Update $\lambda^{(q)} = \min \left\{ \kappa \lambda^{(q-1)}, \lambda_{\text{max}} \right\}$; Update $q\leftarrow q+1$
	\UNTIL Convergence of $\boldsymbol{\theta}_t$
	\STATE Obtain $\boldsymbol{\theta}_{t}^{*} = \boldsymbol{\theta}_{t}^{(q-1)}$ and $\left\{{\boldsymbol{\theta}}^{(i)},{\rm\mathbf{W}}^{(i)}\right\} = \left\{{\boldsymbol{\theta}}_t^{*},{\rm\mathbf{W}}_t^{*}\right\}$; Update $i\leftarrow i+1$
\UNTIL Convergence of total solution
\STATE {\bf return} $\left\{{\boldsymbol{\theta}}^{*},{\rm\mathbf{W}}^{*}\right\} = \left\{{\boldsymbol{\theta}}^{(i-1)},{\rm\mathbf{W}}^{(i-1)}\right\}$.
\end{algorithmic}
\end{algorithm}

\section{Performance Evaluation} \label{PER_EVA}

%

We consider Rician channel for all links, i.e., $\mathbf{H}=\sqrt{\frac{1}{PL}}\left( \sqrt{\frac{\rho}{1+\rho}} {\rm\mathbf{H}}_{\text{LoS}} + \sqrt{\frac{1}{1+\rho}} {\rm\mathbf{H}}_{\text{NLoS}} \right)$. $PL=10^{-PL_{dB}/10}$ is the pathloss at $3.5$ GHz frequency \cite{pathloss}, following $PL_{dB} = 38.88 + 22\text{log}_{10}(D)$ (dB), with $D$ as the distance between links of BS-RIS, RIS-UE or BS-UE. Notation $\rho=3$ is the Rician factor. ${\rm\mathbf{H}}_{\text{LoS}}$ indicates the line-of-sight (LoS) channel component, whilst ${\rm\mathbf{H}}_{\text{NLoS}}$ corresponds to non-line-of-sight (NLoS) path, characterized by Rayleigh fading with an exponential distribution having a unit expectation. The simulated scenario is shown in Fig. \ref{scenario}. $d=80$ m, $d_{H}=200$ m are defined as the horizontal distance between BS-RIS and BS-UE, respectively. While, $d_{V}=50$ m indicates the perpendicular distance of BS-RIS. All UL/DL UEs are uniformly distributed in a circular range with a radius of 50 m. $ \varrho_w$, $\varrho_{\theta}$, and $\varrho$ are set to $0.01$. Note that we consider an imperfect RIS with reflection efficiency set to $\beta = \beta_k=0.9,\forall k\in \mathcal{K}$.

\begin{figure}[!t]
	\centering
	\includegraphics[width=3.3in]{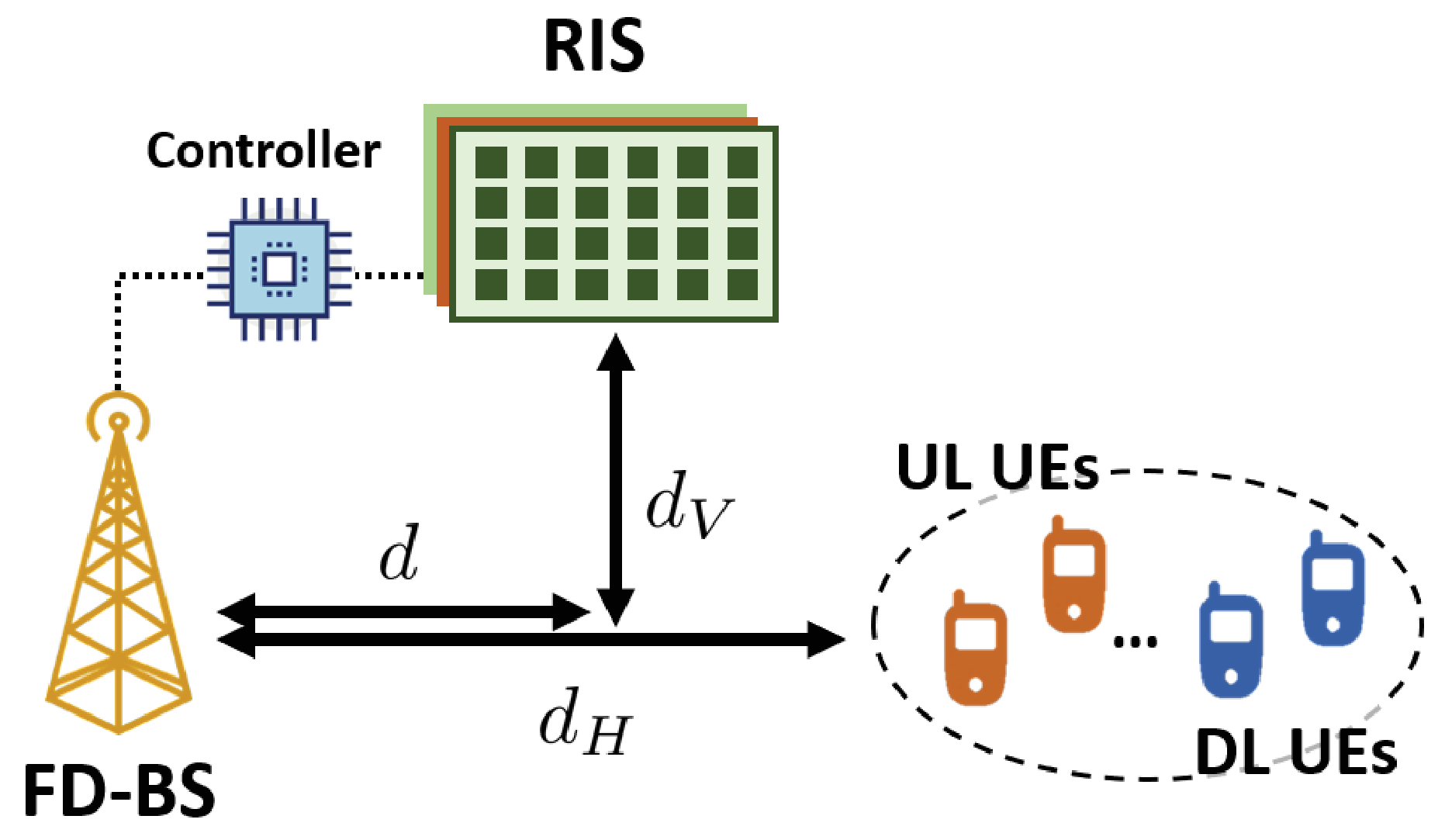}
	\caption{Scenario of RIS deployment.}
	\label{scenario}
\end{figure}
\begin{figure}[!t]
	\centering
	\includegraphics[width=3.3in]{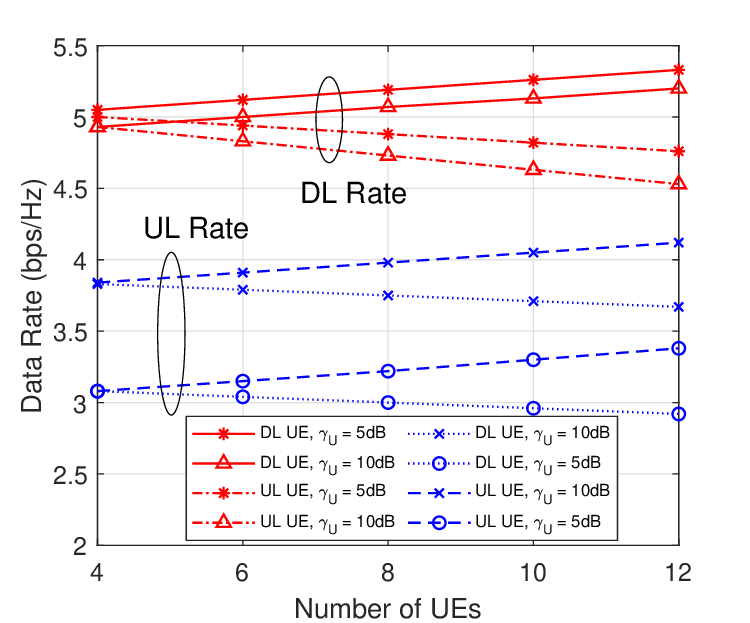}
	\caption{Respective DL/UL Rate w.r.t. different numbers of DL/UL UEs.}
	\label{ULDLrate}
\end{figure}
\begin{figure}[!t]
	\centering
	\includegraphics[width=3.3in]{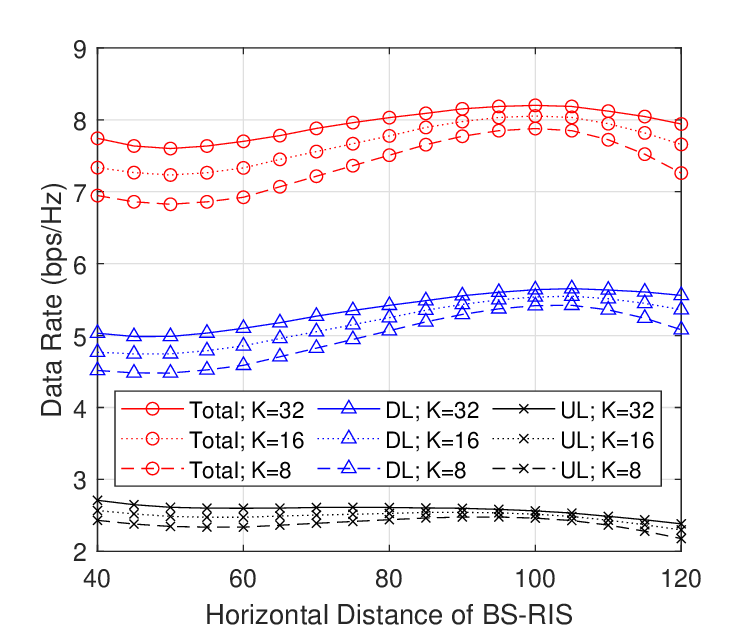}
	\caption{Data rates versus different distances between BS-RIS.}
	\label{distance}
\end{figure}
\begin{figure}[!t]
	\centering
	\includegraphics[width=3.3in]{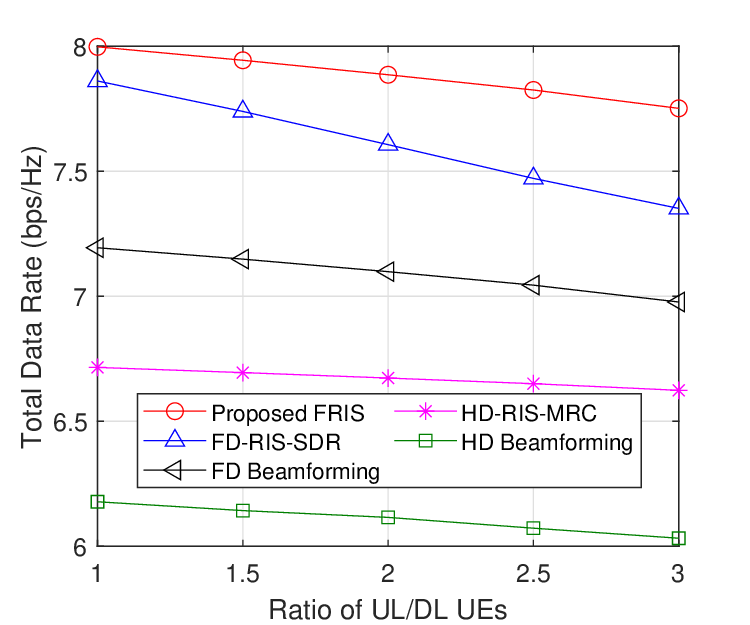}
	\caption{Benchmark comparison.}
	\label{benchmark2}
\end{figure}


Fig. \ref{ULDLrate} depicts respective UL and DL rates of FRIS w.r.t different numbers of DL and UL UEs with $N_t=N_r=6$, $K=16$ and $\gamma_{U} \in \{5, 10\}$ dB. We consider $N=4$ UL UEs when evaluating different numbers of DL UEs, and vice versa for $M=4$ DL UEs in different numbers of UL UEs. We can observe that $\gamma_U=5$ provides higher degree of freedom for DL UEs, corresponding to higher DL rate compared to  $\gamma_U=10$ dB. Also, higher DL rate is achieved with more DL UEs. On the contrary, DL rate decreases under the increasing number of UL UEs owing to severs CCI from UL. Moreover, more DL UEs will induce higher SI reflecting from the RIS, leading to a decrease of UL rate. To elaborate a little further, it exhibits a higher DL rate than UL one, as the system prioritizes the DL objective, while guaranteeing the minimum necessary UL QoS.

Fig. \ref{distance} studies the impact of RIS deployment in FRIS w.r.t. different horizontal distances between BS-RIS $[40,120]$ m, with $N_r=N_t=6$, $M=N=4$ $\gamma_U = 5$ dB, and $K\in \{8,16,32\}$. We can observe three trends when $d \leq 100$ m, $d = 100$ m, and $d \geq 100$ m. When $d \leq 100$, it shows benefits deploy RIS near the BS thanks to more reflected signal power and cancellation of interference. Moreover, they perform a convex shape with minimum rate at $d=50$ m, as UL signals can be weak before it reaches RIS when $d\leq 50$ m. As for $d\in[50,100]$ m, increasing the deployment distance $d$ provides a shorter distance of BS-RIS and RIS-UE, i.e., lower pathloss. The optimal deployment of RIS can be achieved when the distance $d=100$ m between BS-RIS, which corresponds to the midpoint between the BS-RIS and RIS-UE. At this optimal location, it can effectively both mitigate interference and enhance the desired signal power with significant rate improvement. When $d$ goes beyond $100$ m, the pathloss will dominate the result, which results in a decreasing data rate performance. To elaborate a little further, it exhibits a smooth UL rate, as the system is more focused on DL optimization while ensuring the minimum requirement of UL QoS. By appropriately deploying the RIS, it becomes feasible to achieve adequate support for both UL/DL UEs, even with a smaller size of RIS.

In Fig. \ref{benchmark2}, we compare the proposed FRIS scheme to the existing methods: (1) \textbf{FD-RIS-SDR} applies semi-definite relaxation (SDR) in \cite{SDR} for BS beamforming; while RIS part adopts PCCP. (2) \textbf{FD Beamforming} in \cite{FD_oblique} applies a oblique projection technique but without RIS deployment. (3) \textbf{HD-RIS-MRC} designs based on maximum ratio combining (MRC) in both BS/RIS beamforming \cite{BM_MRC} under downlink transmission in HD systems. (4) \textbf{HD Beamforming} relies solely on BS beamforming using \cite{BM_MRC} without the aid of RIS in HD systems. We evaluate different ratios of number of UL UEs to that of DL UEs. We can observe the worst data rate in HD without RIS owing to insufficient channel diversity. With the aid of RIS, HD-RIS-MRC has around $10\%$ rate improvement compared to that without RIS case. Furthermore, FD Beamforming can utilize full spectrum, with enhanced $15\%$ rate compared to that of HD systems without RIS. Upon utilizing SDR approximation in FD-RIS-SDR, it attains a rate improvement of around $6\%$ compared to FD Beamforming. Moreover, the proposed FRIS scheme amends dropping rank-1 issue in SDR, which achieves the highest rate among the other benchmarks.

\section{Conclusions} \label{con}
In this paper, we have proposed FRIS scheme for optimizing active BS and passive RIS beamforming in RIS-enabled FD transmissions. A series of transformation is conducted for obtaining a convex problem as well as dealing with constraint on RIS phase-shifts. In simulations, it strikes a compelling balance between UL/DL rates with different numbers of UEs and QoS. We can also observe the optimal rate from RIS deployment with higher number of RIS elements thanks to increasing channel diversity. Moreover, the proposed FRIS has achieved the highest rate compared to the other approximation method, conventional beamforming techniques, HD systems, and deployment without RIS.

\bibliographystyle{IEEEtran}

\bibliography{IEEEabrv,myReference}

\end{document}